\documentclass[11pt]{article}
\usepackage{fullpage}
\usepackage{graphicx}
\usepackage{amssymb}
\usepackage{amsmath}
\usepackage{amsthm}

\usepackage{stmaryrd}
\usepackage{mathabx}
\usepackage{mathrsfs}
\usepackage{tikz}

\newtheorem{theorem}{Theorem}[section]
\newtheorem{lemma}[theorem]{Lemma}

\newtheorem{corollary}[theorem]{Corollary}

\theoremstyle{definition}
\newtheorem{definition}[theorem]{Definition}
\newtheorem{remark}[theorem]{Remark}
             
\newcommand{\ignore}[1]{}
             
\newcommand{\hcm}[1][1]{\hspace*{#1 cm}}
\newcommand{\rb}[2]{\raisebox{#1 mm}[0mm][0mm]{#2}}
\newcommand{\istrut}[2][0]{\rule[- #1 mm]{0mm}{#1 mm}\rule{0mm}{#2 mm}}
\newcommand{\zero}[1]{\makebox[0mm][l]{$#1$}}

\newcommand{\ceil}[1]{\lceil #1 \rceil}
\newcommand{\floor}[1]{\lfloor #1 \rfloor}

\newcommand{\bydef}{\stackrel{\operatorname{def}}{=}}
\newcommand{\poly}{\operatorname{poly}}

\newcommand{\Furedi}{F\"{u}redi}

\newcommand{\Bona}{B\'{o}na}

\newcommand{\Ex}{\operatorname{Ex}}

\newcommand{\compose}{\operatorname{\circ}}

\newcommand{\lshuffle}{\varolessthan}
\newcommand{\rshuffle}{\varogreaterthan}

\newcommand{\bl}[1]{\llbracket #1 \rrbracket}

\newcommand{\bu}{\bullet}

\newcommand{\Kron}{\otimes}
\newcommand{\oneone}{\mbox{$\left({\tiny\begin{array}{@{}c@{}}
\bu\\
\bu\end{array}}\right)$}}

\newcommand{\hatpattern}{\scalebox{.7}{\mbox{$\left({\tiny\begin{array}{@{}c@{}c@{}c@{}}
&\bu&\\
\bu&&\bu\end{array}}\right)$}}}
\newcommand{\ababa}{\scalebox{.7}{\mbox{$\left({\tiny\begin{array}{@{}c@{}c@{}c@{}c@{}}
\bu&&&\\
&&\bu&\\
&\bu&&\bu\end{array}}\right)$}}}
\newcommand{\ababareflect}{\scalebox{.7}{\mbox{$\left({\tiny\begin{array}{@{}c@{}c@{}c@{}c@{}}
&&&\bu\\
&\bu&\\
\bu&&\bu&\end{array}}\right)$}}}

\newcommand{\threeonetwo}{\scalebox{.7}{\mbox{$\left({\tiny\begin{array}{ccc}
\bu&&\\
&&\bu\\
&\bu&\end{array}}\right)$}}}

\newcommand{\twoonethree}{\scalebox{.7}{\mbox{$\left({\tiny\begin{array}{ccc}
&&\bu\\
\bu&&\\
&\bu&\end{array}}\right)$}}}

\newcommand{\patternA}{\scalebox{.7}{$\left({\tiny\begin{array}{@{}c@{}c@{}c@{}}
&&\bu\\
&\bu&\\
\bu&&\end{array}}\right)$}}
\newcommand{\patternB}{\scalebox{.7}{$\left({\tiny\begin{array}{@{}c@{}c@{}c@{}}
\bu&&\\
&&\bu\\
&\bu&\end{array}}\right)$}}

\newcommand{\Ubot}{U_{\operatorname{bot}}}
\newcommand{\Utop}{U_{\operatorname{top}}}
\newcommand{\Umid}{U_{\operatorname{mid}}}
\newcommand{\Usub}{U_{\operatorname{sub}}}

\newcommand{\Sbot}{S_{\operatorname{bot}}}

\newcommand{\livebl}[1]{\llparenthesis\,{#1}\,\rrparenthesis}

\newcommand{\Afirst}{A_{\operatorname{first}}}
\newcommand{\Alast}{A_{\operatorname{last}}}

\newcommand{\Alight}{A_{\operatorname{light}}}
\newcommand{\Aheavy}{A_{\operatorname{heavy}}}
\newcommand{\Amidmid}{A_{\operatorname{lightmid}}}
\newcommand{\contract}{\operatorname{c}}

\title{Sorting Pattern-Avoiding Permutations\\
via 0--1 Matrices Forbidding Product Patterns\thanks{Supported by NSF grants CCF-1815316 and CCF-2221980 and by the European Research Council (ERC) under the European Union's Horizon 2020 research and innovation programme under grant agreement No 759557.}}

\author{
Parinya Chalermsook\\ Aalto University 
\and 
Seth Pettie\\ University of Michigan
\and
Sorrachai Yingchareonthawornchai\\
Aalto University
}

\date{}

\begin{document}

\maketitle

\begin{abstract}
We consider the problem of comparison-sorting an 
$n$-permutation $S$ that \emph{avoids} some $k$-permutation $\pi$. 
Chalermsook, Goswami, Kozma, Mehlhorn, and Saranurak~\cite{ChalermsookGKMS15}
prove that when $S$ is sorted by inserting the elements into the \textsc{GreedyFuture}~\cite{DemaineHIKP09} binary search tree, the running time is linear in the 
\emph{extremal function} $\Ex(P_\pi\Kron \hatpattern,n)$.
This is the maximum number of 1s in an $n\times n$ 0--1 matrix avoiding $P_\pi \Kron \hatpattern$, where $P_\pi$ is the $k\times k$ permutation matrix of $\pi$, 
and $P_{\pi}\Kron \hatpattern$ is the $2k\times 3k$ Kronecker product of
$P_\pi$ and the ``hat'' pattern $\hatpattern$.  The same time bound
can be achieved by sorting $S$ with Kozma and Saranurak's \textsc{SmoothHeap}~\cite{KozmaS20}.

Applying off-the-shelf results on the extremal functions of 0--1 matrices, it was known that
\[
\Ex(P_\pi\Kron \hatpattern,n) = \left\{\begin{array}{l@{\hcm}l}
\Omega(n\alpha(n)), &\\
O\mathopen{}\left(n\cdot 2^{(\alpha(n))^{3k/2-O(1)}}\right),
\end{array}\right.
\]
where $\alpha(n)$ is the inverse-Ackermann function.   In this paper we 
give \emph{nearly tight} upper and lower bounds on the density of 
$P_\pi\Kron\hatpattern$-free matrices
in terms of ``$n$'', and improve the dependence on ``$k$'' from doubly exponential to singly exponential.
\[
\Ex(P_\pi\Kron \hatpattern,n) = \left\{\begin{array}{l@{\hcm}l}
\Omega\mathopen{}\left(n\cdot 2^{\alpha(n)}\right), &     \mbox{for most $\pi$,}\\
O\mathopen{}\left(n\cdot 2^{O(k^2)+(1+o(1))\alpha(n)}\right),  &   \mbox{for all $\pi$.}
\end{array}\right.
\]
As a consequence, sorting $\pi$-free sequences can be performed 
in $O(n2^{(1+o(1))\alpha(n)})$ time.
For many corollaries of the dynamic optimality conjecture, 
the best analysis uses forbidden 0--1 matrix theory.  
Our analysis may be useful in analyzing
other classes of access sequences on binary search trees.
\end{abstract}

\section{Introduction}

The problem of sorting restricted classes of permutations has been studied for decades.  Knuth~\cite{Knuth73} observed that the class of permutations sortable by a stack is precisely the set of $(2,3,1)$-avoiding permutations; see~\cite{Tarjan72,BiedlGHLM10,MansourSS19,HarjuI01,Elvey-PriceG17a,Elvey-PriceG17b,FelsnerP08,AlbertB15,AtkinsonR02} and \Bona's survey~\cite{Bona02}
for models of restricted sorting devices.
In general, an $n$-permutation $S$ \emph{avoids} a 
$k$-permutation $\pi$ if there do not exist indices 
$i_1 < \cdots < i_k$ for which
\[
\forall p,q\in [k].\;  S(i_p) < S(i_q) \Longleftrightarrow \pi(p) < \pi(q).
\]
In this paper we consider the algorithmic 
problem of comparison-sorting a $\pi$-avoiding $S$.

\paragraph{Decision Tree Complexity.}
Fredman~\cite{Fredman76} observed that if $S$ is known to be
selected from a permutation set $\Gamma$, that $S$ can be sorted
with $O(n + \log |\Gamma|)$ comparisons.  
The \emph{Stanley-Wilf conjecture} (see \Bona~\cite{Bona22}) 
states that if $\Gamma_\pi$ is the set of all $\pi$-avoiding
permutations, that $|\Gamma_{\pi}| \leq (c(\pi))^n$, 
for some constant $c(\pi)$.  This conjecture was reduced
to the \emph{\Furedi-Hajnal conjecture}~\cite{FurediH92} by Klazar~\cite{Klazar00}
and both conjectures were proved
by Marcus and Tardos~\cite{MarcusT04}. 
Together with Fredman~\cite{Fredman76},
this 
implies that the decision-tree complexity of sorting $S$ is 
$O(n\log c(\pi))=O_k(n)$.
Subsequent work has attempted to pin down the leading constant~\cite{Klazar00,MarcusT04,Cibulka09,Fox13,CibulkaK17}.
Fox~\cite{Fox13} proved that\footnote{The manuscript~\cite{Fox13} only gives an $\Omega(k^{1/4}n)$ lower bound on the 
decision tree complexity of sorting a $\pi$-free $S$.  The $\Omega(k^{1/2}n)$ lower bound is unpublished.}
\[
n\log c(\pi) = \left\{\begin{array}{l@{\hcm[.5]}l} 
O(kn)               & \mbox{For all $k$-permutations $\pi$,}\\
\Omega(k^{1/2}n)    & \mbox{For some $k$-permutation $\pi$,}\\
\Omega((k/\log k)^{1/2}n) & \mbox{For almost all $k$-permutations $\pi$.}
\end{array}\right.
\]

\paragraph{Algorithmic Complexity.}
There are two natural ways to approach the \emph{algorithmic} complexity
of sorting a $\pi$-free $S$.  The first is to use knowledge of $\pi$
to structure the sorting process.  
This approach is sufficient to sort optimally in $O(n)$ time when $k=3$~\cite{Knuth73,Arthur07}, 
and has had limited success for some patterns with $k=4$.
Arthur~\cite{Arthur07} gave $O(n)$-time sorting algorithms when
$\pi \in  \{(1,2,3,4), (1,2,4,3), (2,1,4,3)\}$, and 
$O(n\log\log\log n)$-time
sorting algorithms when $\pi \in \{(1,3,2,4),(1,3,4,2),(1,4,2,3),(1,4,3,2)\}$.
The \emph{oblivious} approach to sorting $S$ is to 
simply use a general-purpose sorting algorithm, but analyze its
behavior when $S$ happens to be $\pi$-free.
This is the approach taken by Chalermsook, Goswami, Kozma, Mehlhorn, and Saranurak~\cite{ChalermsookGKMS15}, Kozma and Saranurak~\cite{KozmaS20}, 
and by our paper.
Consider these two general-purpose sorting algorithms:
\begin{description}
\item[\textsf{BST Sort}.] Fix some dynamic binary search tree (BST) algorithm $\mathcal{T}$.  Beginning from an empty BST, 
insert the elements $S(1),\ldots,S(n)$ in that order, reorganizing
the tree between inserts as $\mathcal{T}$ dictates.
The number of comparisons is the sum of depths of $(S(i))_{1\leq i\leq n}$ at the time of their insertion; 
the \emph{time} is linear in the number of comparisons and that needed to reorganize the tree via rotations.

\item[\textsf{Heap Sort}.] Fix some heap data structure $\mathcal{H}$.
Insert the elements $S(1),\ldots,S(n)$ into the heap in that order,
then perform $n$ \textsf{Delete-Min} operations, thereby sorting the sequence.
\end{description}

Chalermsook et al.~\cite{ChalermsookGKMS15} 
analyzed the performance of 
\textsf{BST Sort} when $\mathcal{T}$ is
\textsc{GreedyFuture}~\cite{DemaineHIKP09}, 
an online BST that is $O(1)$-competitive
with the natural offline \textsc{Greedy} algorithm~\cite{Lucas88,Munro00}.
Define $A_S$ to be the $n\times n$ 0--1 permutation matrix where $A_S(i,S(i))=1$.
If $S$ avoids a $k$-permutation $\pi$, 
then $A_S$ is $P_\pi$-free, where $P_\pi(i,\pi(i)) = 1$.
Define $A_{\textsc{Greedy}(S)}(i,j)=1$
iff the element with rank $j$ is touched by the insertion of $S(i)$.
Chalermsook et al.~\cite{ChalermsookGKMS15} proved 
that any occurrence of the 
``hat'' pattern $\hatpattern$ in $A_{\textsc{Greedy}(S)}$ contains,
within its bounding box, an input point of $A_S$, and as a consequence, 
$A_{\textsc{Greedy}(S)}$ avoids $Q = P_\pi \Kron \hatpattern$, 
where $\Kron$ is the Kronecker product, i.e., 
each 1 of $P_\pi$ is replaced by $\hatpattern$.
(Following convention, 0--1 matrices are depicted 
with blanks for 0s and bullets for 1s.  
See Section~\ref{sect:forbidden-matrix-prelims} 
for explicit definitions regarding 0--1 matrices.)
For example, if $\pi=(1,3,2,4)$, ordering rows
from bottom to top:
\begin{align*}
    P_\pi &= \scalebox{.8}{$\left(\begin{array}{cccc}
    &&&\bu\\
    &\bu&&\\
    &&\bu&\\
    \bu&&&
    \end{array}\right)$}
    &
    Q = P_\pi \Kron \hatpattern &= \scalebox{.8}{$\left(\begin{array}{c@{\hcm[.1]}c@{\hcm[.1]}c@{\hcm[.2]}c@{\hcm[.1]}c@{\hcm[.1]}c@{\hcm[.4]}c@{\hcm[.1]}c@{\hcm[.1]}c@{\hcm[.2]}c@{\hcm[.1]}c@{\hcm[.1]}c}
    &&&&&&&&&&\bu&\\
    &&&&&&&&&\bu&&\bu\\
    &&&&\bu&&&&&&&\\
    &&&\bu&&\bu&&&&&&\\
    &&&&&&&\bu&&&&\\
    &&&&&&\bu&&\bu&&&\\
    &\bu&&&&&&&&&&\\
    \bu&&\bu&&&&&&&&&
    \end{array}\right)$}
\end{align*}
If $X$ is a fixed 0--1 pattern matrix, define 
$\Ex(X,n)$ be the maximum number of 1s in an 
$n\times n$ matrix that avoids $X$. 
Thus, the running time of~\cite{ChalermsookGKMS15}
can be bounded in terms of $\Ex(Q,n)$ without knowing 
exactly what it is.

\begin{theorem}[Chalermsook, Goswami, Mehlhorn, Kozma, and Saranurak~\cite{ChalermsookGKMS15}]\label{thm:CGMKS}
If $S$ is $\pi$-free, \textsf{BST Sort} using \textsc{GreedyFuture} 
sorts $S$ in $O(\Ex(Q,n))$ time,
where $Q=P_\pi\Kron \hatpattern$.
\end{theorem}

Observe that $Q$ is a $2k\times 3k$ 
\emph{light} pattern: 
it contains exactly one 1 per column.
There is a well known connection between light patterns and generalized Davenport-Schinzel sequences~\cite{Klazar92,FurediH92,Keszegh09,Pettie-GenDS11,Pettie15-SIDMA}.
Applying a simplifying transformation that collapses the first two rows
\cite[Thm. 2.2]{FurediH92} 
and then \cite[Thm. 1.3]{Pettie15-SIDMA}, 
we have the following general upper bound, 
where $\alpha(n)$ is the inverse-Ackermann function.
\begin{equation}\label{eqn:general-upper-bound}
\Ex(Q,n) \leq \left\{
\begin{array}{l@{\hcm[1]}l}
2n\alpha(n)+O(n)      & k=2\\
n\cdot 2^{(1+o(1))\alpha^t(n)/t!} 
    & \mbox{$k$ odd, $t=(3k-5)/2$}\\
n\cdot (\alpha(n))^{(1+o(1))\alpha^{t}(n)/t!} 
    & \mbox{$k$ even, $t=(3k-6)/2$}
\end{array}
\right.
\end{equation}
Thus, by Theorem~\ref{thm:CGMKS}, \textsc{GreedyFuture} 
sorts $S$
in $O(n\cdot 2^{\alpha(n)^{3k/2-O(1)}})$ time.  
On the lower bound side, we know that $\Ex(Q,n)=\Omega(n\alpha(n))$
as every $Q$ contains one of the two patterns shown below,
which are 
associated with order-3 Davenport-Schinzel sequences~\cite{HS86,FurediH92}.
\[
\scalebox{.7}{$\left(\begin{array}{cccc}
\bu&&&\\
&&\bu&\\
&\bu&&\bu\end{array}\right)$}
\hcm
\scalebox{.7}{$\left(\begin{array}{cccc}
&&&\bu\\
&\bu&\\
\bu&&\bu\end{array}\right)$}
\]

\medskip 

The \textsc{Greedy} algorithm is theoretically attractive, 
but cumbersome to implement online as \textsc{GreedyFuture}~\cite{DemaineHIKP09}.
Kozma and Saranurak~\cite{KozmaS20} introduced 
a new heap data structure called a \textsc{SmoothHeap},
and proved \textsf{Heap Sort} with \textsc{SmoothHeap}
is \emph{equivalent} to \textsf{BST Sort} with \textsc{Greedy}.
Moreover, \textsc{SmoothHeap} is ``naturally'' an online algorithm, 
and is easier to implement than \textsc{GreedyFuture}.
One can define an $n\times n$ 0--1 matrix $A_{\textsc{SmoothHeap}(S)}$ in the same way, where $A_{\textsc{SmoothHeap}(S)}(i,j)=1$ 
iff the $i$th Delete-Min touches the element with rank $j$.  
It is proved that $A_{\textsc{SmoothHeap}}$ avoids a matrix equivalent to $Q$.\footnote{Strictly speaking the equivalence between \textsc{Greedy} and \textsc{SmoothHeap} swaps the roles of time and space.
Sorting $S$ with \textsc{Greedy} is isomorphic 
to sorting $S^T$ with \textsc{SmoothHeap}, 
where $S^T$ is the transpose permutation: $S(i)=j \Leftrightarrow S^T(j)=i$.
Note that $S^T$ avoids $\pi^T$.  Since the extremal functions for $Q$ and $Q^T$ are identical on square matrices, we infer
that the time to sort $S^T$ is also $O(\Ex(Q,n))$.}

\begin{theorem}[Kozma and Saranurak~\cite{KozmaS20}]\label{thm:KS}
If $S$ is $\pi$-free, \textsf{Heap Sort} 
using \textsc{GreedyFuture} sorts $S$ in $O(\Ex(Q,n))$ time,
where $Q=P_\pi \Kron \hatpattern$.
\end{theorem}

The main outstanding question is whether it is possible
to sort in $O_k(n)$ time, and in particular, whether
the \textsc{Greedy}- or \textsc{SmoothHeap}-based 
algorithms of ~\cite{ChalermsookGKMS15,KozmaS20}
already sort in time $O_k(n)$.
It would also be interesting to give a non-trivial upper bound
on the complexity of \textsf{BST Sort} with a \textsc{Splay Tree}~\cite{SleatorT85},
or \textsf{Heap Sort} with a \textsc{Pairing Heap}~\cite{F+86}.

\subsection{New Results}

\subsubsection{Upper Bounds}

Our main result is a new upper bound on the extremal function of $P_\pi \Kron\hatpattern$-type matrices
that has a much weaker dependence on $k$,
which immediately gives better upper bounds on 
the complexity of sorting $\pi$-free sequences via~\cite{ChalermsookGKMS15,KozmaS20}.

\begin{theorem}\label{thm:main}
Let $P_\pi$ be the $k\times k$ permutation matrix
of $\pi$ and $Q = P_\pi \Kron\hatpattern$ be 
a $2k\times 3k$ light matrix.
Then
\[
\Ex(Q,n) \leq n\cdot \left(2^{O(k^2)} + O(\alpha(n))^{3k-2}\right)2^{\alpha(n)} 
    = n\cdot 2^{O(k^2) + (1+o(1))\alpha(n)}.
\]
\end{theorem}

\begin{corollary}\label{cor:main}
If $S$ is $\pi$-free, then \textsf{BST Sort} 
using \textsc{GreedyFuture} and \textsf{Heap Sort}
using the \textsc{SmoothHeap} 
will sort $S$ in $n\cdot 2^{O(k^2) + (1+o(1))\alpha(n)}$ time.
\end{corollary}

One can view 
Corollary~\ref{cor:main} as improving on 
the $n2^{\alpha(n)^{3k/2-O(1)}}$ 
bound of (\ref{eqn:general-upper-bound})
in two ways.  
It is an asymptotic improvement in $n$ as it brings the exponent of $\alpha(n)$ from $3k/2-O(1)$ down to $1$.
However, even if one is tempted to consider $\alpha(n)$ to be a small constant, it also reduces the dependency on $k$ from \emph{doubly} exponential to merely \emph{singly} exponential. 

It is possible to improve the factor $2^{\alpha(n)}$  for a specific product pattern. For example, 
\begin{theorem}\label{thm:identity-product}
If $I_k$ is the $k \times k$ identity matrix, then  
$$\Ex(I_k\Kron\hatpattern,n) \leq 2(k-1)n\alpha(n)+O(kn).$$
\end{theorem}

\subsubsection{Lower Bounds}

When $k\geq 2$, all $Q=P_\pi \Kron\hatpattern$ patterns contain $\ababa$ or its reflection, 
which is known to have extremal function 
$\Ex(\ababa,n) = 2n\alpha(n)\pm O(n)$~\cite{HS86,FurediH92,Nivasch10,Pettie-DS-JACM}.

We prove that $\Ex(P_\pi\Kron\hatpattern,n)=\Omega(n2^{\alpha(n)})$ 
whenever $\pi$ contains $(3,1,2)$ or $(2,1,3)$,
or equivalently, when $P_\pi$ contains 
$\threeonetwo$ or $\twoonethree$.
Thus, Theorem~\ref{thm:W} implies that the general upper bound of Theorem~\ref{thm:main}
can only be improved in the $\poly(\alpha(n))$ factor.

\begin{theorem}\label{thm:W}
$\Ex(W,n) = \Theta(n2^{\alpha(n)})$, where
\[
W = \scalebox{.7}{$\left(\begin{array}{ccccc}
\bu&&&&\\
&&&& \bu\\
&&\bu&&\\
&\bu&&\bu&
\end{array}
\right)$.}
\]
\end{theorem}

\subsection{Pattern-avoidance and the Dynamic Optimality Conjecture}  

The original \emph{dyanamic optimality conjecture}~\cite{SleatorT85}
states that the (online) \textsc{Splay} BST is $O(1)$-competitive
with the optimum offline BST, for any sequence with length $\Omega(n)$.
Today dynamic optimality usually refers to the conjecture that \emph{there exists} an $O(1)$-competitive BST, with \textsc{Greedy} / \textsc{GreedyFuture}~\cite{Lucas88,Munro00,DemaineHIKP09} and \textsc{Splay} being the foremost candidates.

It is an open problem to prove $o(\log n)$-competitiveness for \textsc{Splay} or \textsc{Greedy}, though some \emph{corollaries} of 
dynamic optimality have been proved~\cite{SleatorT85,Tar85,Cole00,ColeEtal00,IaconoL16,ChalermsookGJAPY23,LevyT19}.  Many corollaries of dynamic optimality 
can be characterized by forbidden patterns.
For example, the optimum BST executes all 
of these sequences in linear time.\footnote{Note that the problem we study in this paper is not on this list.  There is no published proof yet to the effect that ``accessing a $\pi$-avoiding $S$ in $O_\pi(n)$ time'' is a corollary of dynamic optimality.  The $O_\pi(n)$-height decision-tree implied by Fredman~\cite{Fredman76} is not obviously implementable as a dynamic binary search tree.}  
\begin{description}
    \item[Sequential.] The sequential access sequence $S=(1,2,\ldots,n)$ avoids $(2,1)$.
    \item[Deque.] In a deletion-only deque sequence, $S(i)$ is either the minimum or maximum of $\{S(i),S(i+1),\ldots,S(n)\}$.  Deque sequences avoid $\{(213),(312)\}$.  (In a deque, the
    accessed elements are also typically deleted from the tree~\cite{Sundar92,Pettie-Deque-08}.)
    \item[Preorder and Postorder.] Let $R$ be any BST over $\{1,\ldots,n\}$ and $S$ be a preorder (or postorder) traversal
    of $R$.  Then $S$ avoids $(231)$ (or $(312)$).
    (The \emph{Traversal Conjecture} of Sleator and Tarjan~\cite{SleatorT85} concerned preorder sequences. If the accessed elements in a preorder sequence are moved to the root and \emph{deleted}, yielding two trees, this corresponds with Lucas's definition of \emph{Split}-sequences~\cite{Lucas91}.)
    \item[$k$-Increasing.] $S$ can be decomposed into $(k-1)$ increasing subsequences, or equivalently, 
    $S$ avoids $(k,\ldots,2,1)$.
    \item[$k$-Recursively Decomposable.] A permutation $S$ is 
    $k$-recursively decomposable if 
    (i) the 1s of the corresponding permutation matrix $A_S$
    can be partitioned into $k$ non-overlapping rectangles, and 
    (ii) those rectangles are themselves $k$-recursively decomposable,
    where in the base case, any $1\times 1$ matrix is $k$-recursively 
    decomposable.  These sequences avoid all \emph{simple} $(k+1)$-permutations.\footnote{A $(k+1)$-permutation $\pi$ is \emph{simple} if there is no interval $I\subset \{1,\ldots,k+1\}$ with $|I|\in [2,k]$ such that $\pi(I) \bydef \{\pi(j) \mid j\in I\} = I$.}
\end{description}

Figure~\ref{fig:relation} shows the relationship between the classes of permutations, and Table~\ref{table:bounds} gives some known upper bounds on the performance of \textsc{Splay} and \textsc{Greedy}. In particular, our new upper bound on $\Ex(P_{\pi}\Kron\hatpattern,n)$ improves on the bounds for $k$-recursively decomposable sequences (when preprocessing is not allowed), and 
$k$-permutation avoiding sequences.

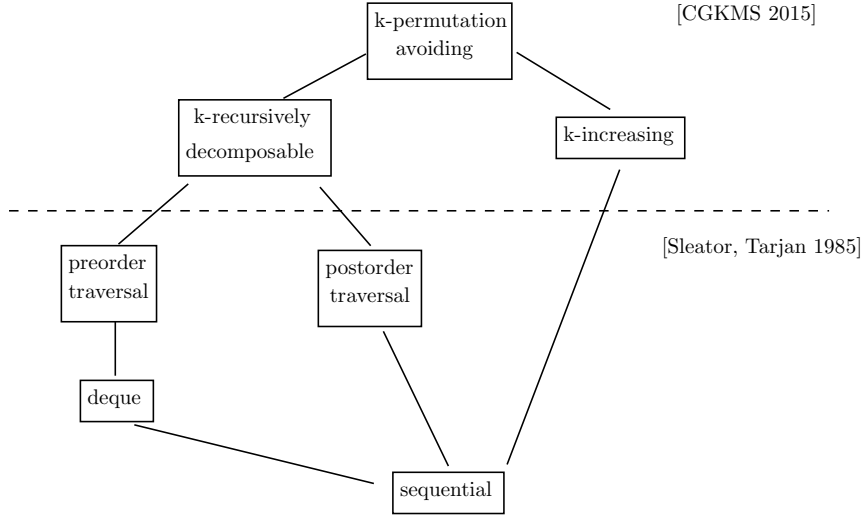
\begin{figure} 
\centering
\tikzset{every picture/.style={line width=0.75pt}} 

\scalebox{.8}{
\begin{tikzpicture}[x=0.75pt,y=0.75pt,yscale=-1,xscale=1]

\draw    (253,41.05) -- (201,69.05) ;
\draw    (348,40.05) -- (407,76.05) ;
\draw    (141,123.05) -- (97,161.05) ;
\draw    (224,125.05) -- (256,162.05) ;
\draw    (95,210.05) -- (95,244.05) ;
\draw    (264,216.05) -- (305,301.05) ;
\draw    (107,275.05) -- (258,312.05) ;
\draw    (413,114.05) -- (342,300.05) ;
\draw  [dash pattern={on 4.5pt off 4.5pt}]  (28,140) -- (545,140) ;

\draw    (254,9) -- (345,9) -- (345,57) -- (254,57) -- cycle  ;
\draw (257,13) node [anchor=north west][inner sep=0.75pt]   [align=left] {{\small k-permutation}\\{\small  \ \ \ avoiding}};
\draw    (135,70) -- (231,70) -- (231,118) -- (135,118) -- cycle  ;
\draw (138,74) node [anchor=north west][inner sep=0.75pt]   [align=left] {\begin{minipage}[lt]{62.43pt}\setlength\topsep{0pt}
\begin{center}
{\small k-recursively }
\end{center}
{\small decomposable}
\end{minipage}};
\draw    (373,81) -- (454,81) -- (454,107) -- (373,107) -- cycle  ;
\draw (376,85) node [anchor=north west][inner sep=0.75pt]   [align=left] {{\small k-increasing}};
\draw    (61,162) -- (121,162) -- (121,210) -- (61,210) -- cycle  ;
\draw (64,166) node [anchor=north west][inner sep=0.75pt]   [align=left] {{\small preorder}\\{\small traversal}};
\draw    (223,165) -- (288,165) -- (288,213) -- (223,213) -- cycle  ;
\draw (226,169) node [anchor=north west][inner sep=0.75pt]   [align=left] {\begin{minipage}[lt]{41.51pt}\setlength\topsep{0pt}
\begin{center}
{\small postorder}\\{\small traversal}
\end{center}

\end{minipage}};
\draw    (73,247) -- (119,247) -- (119,273) -- (73,273) -- cycle  ;
\draw (76,251) node [anchor=north west][inner sep=0.75pt]   [align=left] {{\small deque}};
\draw    (270,305) -- (340,305) -- (340,331) -- (270,331) -- cycle  ;
\draw (273,309) node [anchor=north west][inner sep=0.75pt]   [align=left] {{\small sequential}};
\draw (447,8) node [anchor=north west][inner sep=0.75pt]   [align=left] {{\small [CGKMS 2015]}};
\draw (438,155) node [anchor=north west][inner sep=0.75pt]   [align=left] {{\small [Sleator, Tarjan 1985]}};
\end{tikzpicture}}
\caption{Relation between classes of search sequences. The upper class contains the lower ones. } 
\label{fig:relation}
\end{figure} 

\begin{table}
\centering\scalebox{.95}{\begin{tabular}{|l|l|l|l|l|}
\multicolumn{1}{l}{\textsf{Search Sequence}}
& \multicolumn{1}{l}{\textsf{Forbidden Pattern}}
& \multicolumn{1}{l}{\textsf{Splay}}         
& \multicolumn{1}{l}{\textsf{Greedy}}                                                 
& \multicolumn{1}{l}{\textsf{Citation}} \\ \hline
Sequential & $(21)$-free & $O(n)$      & $O(n)$     & \cite{Tar85}\\ \hline
Deque      & $\{(213),(231)\}$-free    & $O(n \alpha^*(n))$ & $O(n \alpha (n)) $  &  \cite{Pettie-Deque-08,ChalermsookGJAPY23}\\ \hline
Preorder   & $(231)$-free  & ---       & $O(n 2^{\alpha(n)})$  & \cite{ChalermsookGJAPY23}\\ \hline
Postorder  & $(312)$-free  & ---       & $O(n)$  & \cite{ChalermsookGJAPY23}\\ \hline
$k$-Increasing & $(k, \ldots,2,1)$-free & --- & $O(\min \{nk^2, n k \alpha(n)\})$ & \cite{ChalermsookGJAPY23} \\ \hline
$k$-Recursively & avoids all \emph{simple}   &      & $O(n\log k)$ &\\
decomposable    & $(k+1)$-permutations & \rb{2.5}{---}  & (\emph{prepr.~initial tree}) & \rb{2.5}{\cite{GoyalG19}} \\ \hline
$k$-Permutation &             &      & &\\
avoiding        & \rb{2.5}{$\pi$-free}  & \rb{2.5}{---}  & \rb{2.5}{$O(\Ex(P_{\pi}\Kron\hatpattern,n))$} & \rb{2.5}{\cite{ChalermsookGKMS15}}\\ \hline
\end{tabular}}
\caption{\label{table:bounds}Upper Bounds on Structured Search Sequences}
\end{table}

\subsection{Organization}

In Section~\ref{sect:forbidden-matrix-prelims} we
review forbidden 0--1 matrix terminology, and some 
key results.  In Section~\ref{sect:upper-bound} we 
prove Theorem~\ref{thm:main}, establishing 
the $n2^{(1+o(1))\alpha(n)}$ upper bound 
on $P\Kron\hatpattern$-type matrices.
In Section~\ref{sect:lower-bound} we prove 
Theorem~\ref{thm:main-lower-bound}'s $\Omega(n2^{\alpha(n)})$
lower bound on $W$-free matrices.
Section~\ref{sect:additional-upper-bounds} presents
some additional upper bounds, on $I_k\Kron\hatpattern$-free matrices (Theorem~\ref{thm:identity-product}) and 
matrices avoiding $W$ and its reflection.
We conclude with some open problems in Section~\ref{sect:conclusion}.

\section{Preliminaries}\label{sect:forbidden-matrix-prelims}

Let $A\in\{0,1\}^{n\times m}$ 
and $P\in\{0,1\}^{k\times l}$.  The \textit{weight} of $A$, denoted as $\|A\|_1$, is the number of 1's in $A$. 
We say $P$ is \textit{contained} in $A$, written $P\prec A$ 
if there are row indices $r_1 < \cdots < r_k$
and column indices $c_1 < \cdots < c_l$ 
such that $P(i,j)=1 \rightarrow A(r_i,c_j)=1$.
In other words, you can obtain $P$ from $A$ 
by deleting rows and columns, and flipping 
some 1s to 0.
The extremal functions are defined as follows.
\begin{align*}
\Ex(P,n,m) &= \max\{\|A\|_1 \mid A\in\{0,1\}^{n\times m}, P\nprec A\},\\
\Ex(P,n) &= \Ex(P,n,n).
\end{align*}

If $P$ is a $k\times k$ 
permutation matrix, 
it is known that both $\Ex(P,n)$ 
and $\Ex(P\Kron \oneone,n)$ are $O_k(n)$,
but we will be interested in the leading constants as well.

\begin{theorem}[Marcus and Tardos~\cite{MarcusT04}, Geneson~\cite{Geneson09}, Fox~\cite{Fox13}, Cibulka and Kyncl~\cite{CibulkaK17}, Geneson~\cite{GenesonPhD15}, Geneson and Tian~\cite{GenesonT17}]\label{thm:perm-doubleperm-constant}
Let $P$ be any permutation matrix. Then there exists constants $C_k,C_k' \leq 2^{(4+o(1))k}$ such that
\begin{align*}
    \Ex(P,n,m) &\leq C_k(n+m),\\
    \Ex(P\Kron \oneone, n,m) &\leq C_k'(n+m).
\end{align*}
\end{theorem}

\section{The Upper Bound}\label{sect:upper-bound}

\subsection{Establishing the General Recurrence}

Let $P$ be a $k\times k$ permutation matrix and
$Q = P \Kron \hatpattern$ 
be the $2k\times 3k$ forbidden pattern. 
Define $Q_{a,b}$ to be the 
$2k\times (3k-(a+b))$ matrix derived from $Q$ 
by removing the first $a$ and last $b$ columns.  
For reasons that will become clear later, 
we must redefine the \emph{contains} relation 
$\prec$ differently for the $Q_{a,b}$ matrices.

\begin{definition}\label{def:Qprec}
We will say that $Q_{a,b}\prec A$ 
if there are 
$2k$ rows $r_1 < \cdots < r_{2k}$ and 
$3k-a-b$ columns $c_1< \cdots < c_{3k-a-b}$
such that
\begin{itemize}
\item $Q_{a,b}(i,j)=1$ implies $A(r_i,c_j) = 1$
\item If $\forall j.\; Q_{a,b}(i,j)=0$ 
then $\exists j'.\; A(r_i,j')\neq 0$.  
In other words, an all-0 row $Q_{a,b}(i,\cdot)$ 
cannot match an all-0 row of $A$.
(Note that $j'$ need not be in $\{c_1,\ldots,c_{3k-a-b}\}$.)
\end{itemize}
\end{definition}

Let $A$ be an $n\times m$ $Q_{a,b}$-free 
matrix with weight $\Ex(Q_{a,b},n,m)$.  
We will classify all 1s in $A$ according to the 
following taxonomy, and bound the number of 1s
in each class directly or inductively.

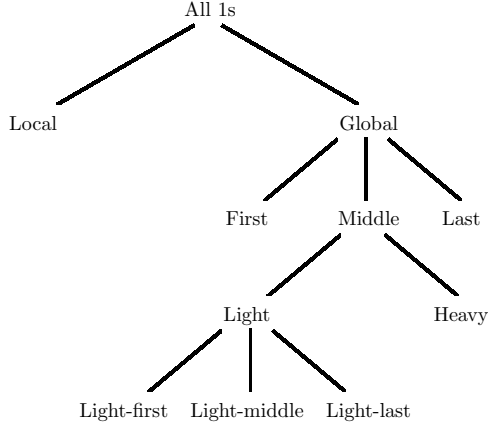
\centerline{
\scalebox{.7}{
\begin{tabular}{cccccc}
&&\\
\multicolumn{4}{c}{All 1s}\\
\zero{\hcm[.5]\rotatebox{30}{\rule{3cm}{0.5mm}}}
\zero{\hcm[3.5]\reflectbox{\rotatebox{30}{\rule{3cm}{0.5mm}}}}
\\
Local &&& Global\\
\zero{\hcm[4.5]\rotatebox{40}{\rule{1.8cm}{0.5mm}}}
\zero{\hcm[6.3]\rotatebox{90}{\rule{1.2cm}{0.5mm}}}
\zero{\hcm[6.6]\reflectbox{\rotatebox{40}{\rule{1.8cm}{0.5mm}}}}
\\
&& First & Middle & Last\\
\zero{\hcm[4.5]\rotatebox{40}{\rule{1.8cm}{0.5mm}}}
\zero{\hcm[6.6]\reflectbox{\rotatebox{40}{\rule{1.8cm}{0.5mm}}}}\\
&& Light && Heavy\\
\zero{\hcm[2.3]\rotatebox{40}{\rule{1.8cm}{0.5mm}}}
\zero{\hcm[4.1]\rotatebox{90}{\rule{1.2cm}{0.5mm}}}
\zero{\hcm[4.4]\reflectbox{\rotatebox{40}{\rule{1.8cm}{0.5mm}}}}
\\
& Light-first & Light-middle & Light-last\\
&&\\
\end{tabular}
}
}

Partition $A$ into \emph{slabs} of $B$ consecutive columns.
A row is called \emph{local} if it has a non-zero 
intersection with exactly one slab and \emph{global} 
otherwise.  The 1s in local/global rows are themselves local/global.
Let $n_i$ be the number of rows local to slab $i$ and 
$n^*$ be the number of global rows, so $n=n^*+\sum_i n_i$.

Suppose $A(r,c)=1$ is a 1 appearing in a global row $r$
and slab $s = \ceil{c/B}$.  We classify this 1
as \emph{first} if the intersection of 
row $r$ and slabs $1,\ldots,s-1$ are zero,
\emph{last} if the intersection of row $r$ and slabs $s+1,\ldots,\ceil{m/B}$ is zero, and \emph{middle} otherwise.

Since each slab is itself 
$Q_{a,b}$-free, the total number 
of local 1s is at most 
\begin{equation}\label{eqn:local}
\sum_{i=1}^{\ceil{m/B}} \Ex(Q_{a,b},n_i,m_i),
\end{equation} 
where $m_i$ is the number of columns
in slab $i$, 
which is exactly $B$ except perhaps the last slab.
Similarly, 
if $\Afirst$ and $\Alast$ are the matrices of 
first 1s and last 1s, then each slab of $\Afirst$ 
is $Q_{a,b+1}$ free, and each slab of $\Alast$ is $Q_{a+1,b}$-free; see Figure~\ref{fig:first-ones}. 
Letting $n_i^f$ ($n_i^l$) be the number of rows
with first (last) 1s in slab $i$, we can
upper bound first and last 1s as follows.

\begin{align}
    \|\Afirst\|_1+\|\Alast\|_1 
    &\leq \sum_{i=1}^{\ceil{m/B}} \left(\Ex(Q_{a,b+1},n_i^f,m_i) + \Ex(Q_{a+1,b},n_i^l,m_i)\right)\nonumber\\
    &\leq \Ex(Q_{a,b+1},n^*,m-m_{\ceil{m/B}}) + \Ex(Q_{a+1,b},n^*,m-m_1).\label{eqn:firstlast}
\end{align}

\begin{figure}
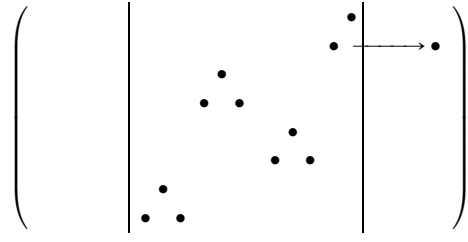

    \centering
    \[
    \scalebox{.8}{$\left(\begin{array}{c@{\hcm[1.5]}|c@{\hcm[.1]}c@{\hcm[.1]}c@{\hcm[.2]}c@{\hcm[.1]}c@{\hcm[.1]}c@{\hcm[.4]}c@{\hcm[.1]}c@{\hcm[.1]}c@{\hcm[.2]}c@{\hcm[.1]}c@{\hcm[.1]}|@{\hcm[1.1]}c}
    &&&&&&&&&&&\bu&\\
    &&&&&&&&&&\bu&\zero{\xrightarrow{\hcm}}&\bu\\
    &&&&&\bu&&&&&&&\\
    &&&&\bu&&\bu&&&&&&\\
    &&&&&&&&\bu&&&&\\
    &&&&&&&\bu&&\bu&&&\\
    &&\bu&&&&&&&&&&\\
    &\bu&&\bu&&&&&&&&&
    \end{array}\right)$}
    \]
    \caption{Vertical lines mark the boundary of some slab.  
    If $Q_{0,1}$ appears in one slab of $\Afirst$, 
    then there must be an occurrence of $Q=Q_{0,0}$ in $A$.}
    \label{fig:first-ones}
\end{figure}

In Eqn.~(\ref{eqn:firstlast}) 
we use the \emph{superadditivity} of $\Ex$ to simplify
the expression.  For any $R$, 
$\Ex(R,n_1,m_1)+\Ex(R,n_2,m_2)\leq \Ex(R,n_1+n_2,m_1+m_2)$.
Note that $\sum_i n_i^f = \sum_i n_i^l = n^*$ and that the first
and last slabs contain no last 1s and first 1s, respectively.

Let $A^*$ be the $n^* \times m$ matrix formed by 
the global rows and containing only the middle 1s.
We partition the rows of $A^*$ into horizontal slabs
of $G$ rows each, so the intersections of the horizontal and vertical slabs induce $G\times B$ \emph{blocks}.
Call a $G\times B$ block in $A^*$ 
\emph{heavy} if it contains a $\hatpattern$, and \emph{light} otherwise.
The middle 1s inside heavy/light blocks are themselves called heavy/light.
Let $\Aheavy$ and $\Alight$ be the $n^*\times m$ matrices 
containing heavy and light 1s, respectively.
In a light block, the first 1 and last 1 of each row are called \emph{light-first} and \emph{light-last}, and all other 1s in the row are \emph{light-middle}.

Define $\Aheavy^{\contract}$ to be the $n^*/G \times m/B$
matrix obtained by \emph{contracting} each block in $\Aheavy$
to a single entry, i.e., non-zero blocks become 1 and 
all-zero blocks become 0.
Because each heavy block contains
a $\hatpattern$, $\Aheavy^{\contract}$ is $P$-free, 
implying $\|\Aheavy^{\contract}\|_1$ 
(the number of heavy blocks)
is at most $\Ex(P,n^*/G,m/B)$.
Since each heavy block consists solely of middle 1s,
each is $Q_{a+1,b+1}$-free; see Figure~\ref{fig:middle}.  Thus,

\begin{figure}
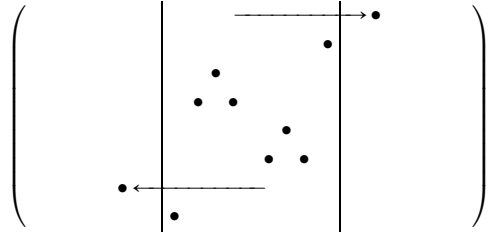

    \centering
\[
\scalebox{.8}{$\left(\begin{array}{@{\hcm[1.5]}c@{\hcm[.5]}|@{\hcm[.1]}c@{\hcm[.2]}c@{\hcm[.1]}c@{\hcm[.1]}c@{\hcm[.4]}c@{\hcm[.1]}c@{\hcm[.1]}c@{\hcm[.2]}c@{\hcm[.1]}|@{\hcm[.5]}c@{\hcm[1.5]}}
    &&&&\zero{\xrightarrow{\hcm[2]}}&&&&&\bu\\
    &&&&&&&&\bu&\\
    &&&\bu&&&&&&\\
    &&\bu&&\bu&&&&&\\
    &&&&&&\bu&&&\\
    &&&&&\bu&&\bu&&\\
    \bu\,\zero{\xleftarrow{\hcm[2]}}&&&&&&&&&\\
    &\bu&&&&&&&&
    \end{array}\right)$}
\]
    \caption{If an instance of $Q_{2,2}$ is contained in a single slab of middle 1s (e.g., $\Aheavy$ or $\Alight$), then $Q_{1,1}$ must also appear in $A$. 
    This inference relies on how \emph{contains}
    is defined for $Q_{a,b}$ matrices in Definition~\ref{def:Qprec}.  In particular, it is critical that all-zero rows of $Q_{2,2}$ must \emph{not} be all-zero in the instance of middle 1s.}
    \label{fig:middle}
\end{figure}

\begin{align}
    \|\Aheavy\|_1 &\leq \Ex(P,n^*/G,m/B)\cdot \Ex(Q_{a+1,b+1},G,B).\label{eqn:heavy}
\end{align}

Let $\Alight^{\contract}$ 
be obtained by contracting the $B$ columns in each 
slab of $\Alight$ to a single column.  
$\Alight^{\contract}$ inherits the $Q_{a,b}$-freeness 
of $\Alight$ and $A$, so the contribution 
of light 1s in the light-first and light-last categories
is at most
\begin{align}
    2\|\Alight^{\contract}\|_1 \leq 2\Ex(Q_{a,b},n^*,m/B).\label{eqn:lightfirstlast}
\end{align}
What remains is to bound the light 1s in the light-middle
category.  Construct an $n^*/G\times m/B$ matrix
$\Amidmid$ by the following procedure, 
which is similar to that of~\cite{Geneson09}.  
Assume the rows of $\Amidmid$ 
are numbered from bottom to top.
For each $i$ independently, 
scan the blocks in slab $i$ that contain light-middle 1s
from bottom to top, setting $\Amidmid(\ell_0,i)=\Amidmid(\ell_1,i)=\cdots =1$
according to the following rules.  See Figure~\ref{fig:chunk}.
\begin{enumerate}
    \item $(\ell_0,i)$ is the first block in slab $i$ containing a light-middle 1.
    \item $\ell_j>\ell_{j-1}$ is the first 
    index such
    that some column in blocks $(\ell_{j-1},i),\ldots,(\ell_j,i)$ 
    contains two light-middle 1s.
\end{enumerate}

\newcommand{\xdownarrow}[1]{%
  {\left\downarrow\vbox to #1{}\right.\kern-\nulldelimiterspace}
}

\begin{figure}
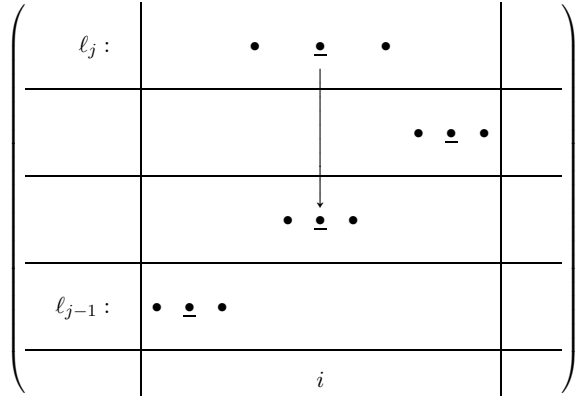

    \centering
\[
\scalebox{.8}{$\left(
\begin{array}{@{\hcm[.5]}r@{\hcm[.5]}|ccccccccccc|@{\hcm[1]}}
        &&&&&&&&&&&\\
\ell_j:  &&&&\bu&&\mbox{\underline{$\bu$}}&&\bu&&&\\
        &&&&&&&&&&&\\\hline
        &&&&&&&&&&&\\
        &&&&&&&&&\bu&\mbox{\underline{$\bu$}}&\bu\\
        &&&&&&&&&&&\\\hline
        &&&&&&&&&&&\\
        &&&&&\bu&\mbox{\underline{$\bu$}}\zero{\hcm[-.23]\rb{13.5}{$\xdownarrow{1.3cm}$}}&\bu&&&&\\
        &&&&&&&&&&&\\\hline
        &&&&&&&&&&&\\
\ell_{j-1}: &\bu&\mbox{\underline{$\bu$}}&\bu&&&&&&&&\\
        &&&&&&&&&&&\\\hline
        &&&&&&i&\istrut{6}&&&&\\
\end{array}
\right)$}
\]
    \caption{Vertical and horizontal lines mark block boundaries.  Underlined 1s are light-middle 1s.}
    \label{fig:chunk}
\end{figure}

Call the interval of 
blocks $(\ell_{j-1},i),\ldots,(\ell_j-1,i)$ 
in $\Alight$ (i.e., excluding $(\ell_j,i)$) 
a \emph{chunk}.  
By construction, 
the intersection of a column and a chunk can contain
at most one light-middle 1.  
(Note that no light block contains two light-middle 1s in the same
column, for otherwise it would contain a $\hatpattern$
pattern and be classified as heavy.)
We claim $\Amidmid$ is $P\Kron \oneone$-free, 
and therefore the number of light-middle 1s in $\Alight$
is, by superadditivity, at most
\begin{align}
    B\cdot \|\Amidmid\|_1 \leq B\cdot \Ex(P\Kron\oneone, n^*/G,m/B) 
    \leq \Ex(P\Kron\oneone, Bn^*/G,m).\label{eqn:lightmiddle}
\end{align}
Consider an occurrence of $\oneone$ in $\Amidmid$,
say $\Amidmid(\ell_j, i) = \Amidmid(\ell_{j'},i)=1$.
By construction they lie in different chunks, thus
there must be a column in slab $i$ of $\Alight$ 
that contains two light-middle 1s in blocks $(\ell_j,i),\ldots,(\ell_{j'},i)$ inclusive. 
Together with a light-first and light-last 1, this forms
a $\hatpattern$ pattern.  Thus, any occurrence
of $P\Kron\oneone$ in $\Amidmid$ implies an occurrence
of $Q = P\Kron\hatpattern$ in $\Alight$,
contradicting the fact that $\Alight$ is $Q_{a,b}$-free.

Combining Eqns.~(\ref{eqn:local},\ref{eqn:firstlast},\ref{eqn:heavy},\ref{eqn:lightfirstlast},\ref{eqn:lightmiddle}), we arrive at a recursive upper bound on $\Ex(Q_{a,b},n,m)$.
\begin{align}
    \Ex(Q_{a,b},n,m) &\leq 
    \sum_{i=1}^{\ceil{m/B}} \Ex(Q_{a,b},n_i,m_i)  & \mbox{local 1s}\nonumber\\
    &\hcm[.5] +\Ex(Q_{a,b+1},n^*,m-m_{\ceil{m/B}}) + \zero{\Ex(Q_{a+1,b},n^*,m-m_1)} \hcm[3.5] & \mbox{first and last 1s}\nonumber\\
    &\hcm[.5] +\Ex(P,n^*/G,m/B)\cdot \Ex(Q_{a+1,b+1},G,B)\hcm[1.2] & \mbox{heavy middle 1s}\nonumber\\
    &\hcm[.5] +2\Ex(Q_{a,b},n^*,m/B) & \mbox{light-first/-last 1s}\nonumber\\
    &\hcm[.5] +\Ex(P\Kron\oneone, Bn^*/G, m).  & \mbox{light-middle 1s}\label{eqn:main-rec}
\end{align}

\subsection{Analysis of The Recurrence}

\begin{lemma}\label{lem:base-case-linear}
Let $t=3k-(a+b)$ be the number of 1s in $Q_{a,b}$.  If $t=3$ then 
$\Ex(Q_{a,b},n,m) \leq 2n + (2k-1)(m-2)$
and if $t=2$ then 
$\Ex(Q_{a,b},n,m) \leq n + (2k-1)(m-1)$.
\end{lemma}

\begin{proof}
First consider $t=3$.
$Q_{a,b}$ contains only three 1s and 
either $2k-2$ or $2k-3$ all-zero rows.  
Those three 1s are equivalent to $\hatpattern$, $\patternA$, or $\patternB$.  Suppose $A$ is $Q_{a,b}$-free.  
Remove the first and last 1 in each row of $A$,
then remove the first $2k-1$ 1s in each of $m-2$ columns,
excluding the first and last, which are now all zero.
If any 1 remains, then there must have been 
an occurrence of $Q_{a,b}$ in $A$.
The $t=2$ case is proved similarly.
\end{proof}

\begin{lemma}\label{lem:base-case-nonlinear}
If $m\leq 2^j$, $\Ex(Q_{a,b},n,m) \leq 2^{t-2}n + (2k-1)j^{\max\{0,t-3\}}(m-2)$, where $t=3k-(a+b)$.
\end{lemma}

\begin{proof}
The cases $t\in\{2,3\}$ follow from Lemma~\ref{lem:base-case-linear},
so we may assume $t > 3$. 
We consider a simplified version of (\ref{eqn:main-rec}) 
in which $B=\ceil{m/2}$, i.e., $m_1=\ceil{m/2}$ and $m_2=\floor{m/2}$.
There are only two slabs, all 1s are classified as \emph{local}, \emph{first}, or \emph{last}, and we have
\begin{align*}
    \Ex(Q_{a,b},n,m) &\leq \sum_{i\in\{1,2\}} \Ex(Q_{a,b},n_i,m_i)
    + \Ex(Q_{a,b+1},n^*,m_1) + \Ex(Q_{a+1,b},n^*,m_2).
\intertext{Applying the inductive hypothesis to each term, this is at most}
&\leq 2^{t-2}(n_1 + n_2) + (2k-1)(j-1)^{t-3}\left(\ceil{m/2}-2+\floor{m/2}-2\right)\\
&\hcm[.5] + 2\cdot 2^{t-3}n^* + (2k-1)(j-1)^{t-4}\left(\ceil{m/2}-2+\floor{m/2}-2\right)\\
&= 2^{t-2}n + (2k-1)\left((j-1)^{t-3} + (j-1)^{t-4}\right)(m-4)\\
&\leq 2^{t-2}n + (2k-1)j^{t-3}(m-2).
\end{align*}
\end{proof}

We use the following version of 
Ackermann's function and its inverses.
\begin{align*}
    a_{1,j} &= 2^j  & \mbox{for $j\geq 1$,}\\
    a_{i,1} &= 2    & \mbox{for $i\geq 2$,}\\
    a_{i,j} &= w\cdot a_{i-1,w}, \; \mbox{where $w = a_{i,j-1}$.} & \mbox{for $i,j\geq 2$,}\\
\ &\ \\
\alpha(n,m) &= \min\{i : a_{i,j}\geq m, \mbox{where $j=\max\{3,\ceil{n/m}\}$}\}\\
\alpha(n) &= \alpha(n,n)
\end{align*}

Observe that in the table of Ackermann values, 
the 1st column is constant ($a_{i,1}=2$)
and the second merely exponential ($a_{i,2}=2^{i+1}$)
so we have to look to the third column to see 
Ackermann-type growth, which is why we set $j$ as 
$j=\max\{3,\ceil{n/m}\}$.

\begin{lemma}\label{lem:mu}
Fix a constant $c = 3k$.
Suppose $m\leq (a_{i,j})^c$.
Then
\[
\Ex(Q_{a,b},n,m)\leq 
\mu_{i,t}(n + (cj)^{\max\{0,t-3\}}(2k-1)(m-2)),
\]
where $t=3k-(a+b)$ and
$\mu_{i,t} = (2^{O(kt)} + O(i)^{t-2})2^{i}$.
\end{lemma}

\begin{proof}
The proof is by induction on $i,j$, and $t$.
The cases $t\in\{2,3\}$ were already handled, so assume $t\geq 4$.
Let $A$ be a $Q_{a,b}$-free $n\times m$ matrix, 
where $m\le (a_{i,j})^c$.
We apply Eqn.~(\ref{eqn:main-rec}) with $B,G$ set as follows:
\begin{align*}
    B &= a_{i,j-1}^c,\\
    G &= (c(j-1))^{\max\{0,t-5\}}(2k-1)(B-2).
\intertext{Observe that}
m/B &\leq (a_{i,j}/a_{i,j-1})^c = (a_{i-1,a_{i,j-1}})^c.
\end{align*}
We apply the induction hypothesis 
at $(i,j-1,t)$ to local 1s
at $(i,j,t-1)$ to first/last 1s,
at $(i,j-1,t-2)$ to heavy middle 1s, 
and at $(i-1,a_{i,j-1},t)$ to light-first/light-last 1s.
Plugging these bounds into Eqn.~(\ref{eqn:main-rec}) 
and applying Theorem~\ref{thm:perm-doubleperm-constant}
yields 
the following upper bound.
\begin{align}
\lefteqn{\Ex(Q_{a,b},n,m)}\nonumber\\
&\leq \mu_{i,t}(n-n^*) + \mu_{i,t}(c(j-1))^{t-3}(2k-1)(m-2m/B) && \mbox{local}\nonumber\\
&\hcm[.4]+ 2\mu_{i,t-1}n^* + 2\mu_{i,t-1}(cj)^{t-4}(2k-1)(m-2) && \mbox{first/last}\nonumber\\
&\hcm[.4]+ C_k(n^*/G + m/B)\left(\mu_{i,t-2}G + \mu_{i,t-2}(c(j-1))^{\max\{0,t-5\}}(2k-1)(B-2)\right) && \mbox{heavy}\nonumber\\
&\hcm[.4]+ 2\mu_{i-1,t}n^* + 2\mu_{i-1,t}(c(a_{i,j-1}))^{t-3}(2k-1)(m/B-2)\nonumber && \mbox{light-first/last}\\
&\hcm[.4]+ C_k'(Bn^*/G + m) && \mbox{light-middle}\nonumber\\
\intertext{Note that by choice of $G$, the line for heavy 1s
is exactly $2C_k\mu_{i,t-2}(n^* + Gm/B)$. Continuing,}
&\leq \mu_{i,t}(n 
+ (cj)^{t-3}(2k-1)(m-2))\label{eqn:jminus1-to-j}\\
&\hcm[.4]+ \left[-\mu_{i,t} + 2\mu_{i,t-1} + 2C_k\mu_{i,t-2} + 2\mu_{i-1,t} + C_k'\right]n^*\label{eqn:n-star-term}\\
&\hcm[.4]+\Big[-\mu_{i,t}c(cj)^{t-4} +  2\mu_{i,t-1}(cj)^{t-4}+ 2C_k\mu_{i,t-2}(c(j-1))^{\max\{0,t-5\}}\label{eqn:m-term}\\ 
&\hcm[1.2]+ 2\mu_{i-1,t}c^{t-3}a_{i,j-1}^{(t-3)-c} + C_k'\Big](2k-1)(m-2)\nonumber\\
&\leq \mu_{i,t}(n + (cj)^{t-3}(2k-1)(m-2)).\label{eqn:complete-induction}
\end{align}

Lines~(\ref{eqn:jminus1-to-j}--\ref{eqn:m-term}) follow 
from the fact that 
$(c(j-1))^{t-3} \leq (cj)^{t-3}-c(cj)^{t-4}$.
Line~(\ref{eqn:complete-induction}) completes
the induction so long as the bracketed
terms in Lines~(\ref{eqn:n-star-term},\ref{eqn:m-term}) are non-positive.  
These will hold whenever
Eqns.~(\ref{constraint:mu-1},\ref{constraint:mu-2}) hold.
\begin{align}
    \mu_{i,t} &\geq 2\mu_{i,t-1} + 2C_k\mu_{i,t-2} + 2\mu_{i-1,t} + C_k',\label{constraint:mu-1}\\
    \mu_{i,t} &\geq \frac{2\mu_{i,t-1}}{c} + \frac{2C_k\mu_{i,t-2}}{c} 
    + \frac{2\mu_{i-1,t}}{2^{3k-1}}\label{constraint:mu-2} + \frac{C_k'}{2^{t-4}c^{t-3}}.
\end{align}
Eqn.~(\ref{constraint:mu-2}) was 
obtained by dividing 
through by $c(cj)^{t-4}$ and noting that $j\geq 2$ and $a_{i,j-1}\geq 2$.  
Clearly any values 
$(\mu_{i,t})_{i\geq 1,t\geq 0}$ that satisfy Eqn.~(\ref{constraint:mu-1}) also satisfy~(\ref{constraint:mu-2}) so we may focus solely on the former.
We argue that the lemma is satisfied for $\mu_{i,t}$
defined as follows.  
Let $C = C_k' \geq C_k$.
\begin{align}
\mu_{i,t} &= (2C + 3i)^{t-2}(2^i-1).
\end{align}
When $t\in\{2,3\}$
the claim follows from Lemma~\ref{lem:base-case-linear} 
since $\mu_{i,3}\geq 2$ and $\mu_{i,2}\geq 1$.
When $i=1$ and $t\geq 4$, 
$m\leq (a_{1,j})^c = a_{1,cj}=2^{cj}$ 
and the claim follows 
from Lemma~\ref{lem:base-case-nonlinear}
since $\mu_{i,t} \geq 2^{t-2}$.
Now suppose $i\geq 2$, $t\geq 4$.
\begin{align*}
    &2\mu_{i,t-1} + 2C_k \mu_{i,t-2} + 2\mu_{i-1,t} + C_k'\\
    &\leq 2(2C+ 3i)^{t-3}(2^i-1)
    + 2C (2C + 3i)^{t-4}(2^i-1)
    + 2(2C+3(i-1))^{t-2}(2^{i-1}-1) + C\\
    &\leq (2C+3i)^{t-2}(2^i-1)
    \left(\frac{2}{2C+3i} + \frac{2C}{(2C+3i)^2} + 1 - \frac{3}{2C+3i}\right)\\
    &\leq (2C+3i)^{t-2}(2^i-1)
    \left(\frac{2}{2C+3i} + \frac{1}{2C+3i} + 1 - \frac{3}{2C+3i}\right)\\
    &\leq (2C+3i)^{t-2}(2^i-1) \; = \, \mu_{i,t}.
\end{align*}
The first inequality is from the inductive hypothesis
and $C_k\leq C_k'\leq C$.
The second inequality follows from
$(2C+3(i-1))^{t-2} \leq (2C+3i)^{t-2} - 3(2C+3i)^{t-3}$.
This completes the induction.
\end{proof}

\begin{proof}[Proof of Theorem~\ref{thm:main}]
Let $A$ be a $Q$-free $n\times m$ matrix
and $t=c=3k$.
Take $i$ to be minimal such that for $j=\max\{3,\ceil{n/m}^{1/t}\}$, $m\leq (a_{i,j})^c$.
It is tedious, but straightforward, 
to show that $i=\alpha(n,m)\pm O(1)$.
Lemma~\ref{lem:mu} bounds the number of 1s in $A$
by
\begin{align*}
    \mu_{i,t}(n + (cj)^{t-3}(2k-1)m)   
    &= \mu_{i,t}\left(n + 2^{O(k\log k)}n\right) & \text{$(cj)^{t-3} < 2^{O(k\log k)}(n/m)$}\\
    &= n\cdot 2^{O(k\log k)}\cdot (2C + 3i)^{t-2}2^i\\
    &= n\cdot \left(2^{O(k^2)} + O(i)^{3k-2}\right)2^i  & \text{\ \ \ \ $C=2^{O(k)}$; see Theorem~\ref{thm:perm-doubleperm-constant}}\\
    &= n\cdot 2^{O(k^2) + (1+o(1))\alpha(n,m)}. 
\end{align*}
\end{proof}

\section{Lower Bounds on 0--1 Matrices via Sequences}\label{sect:lower-bound}

\paragraph{Blocked Sequences and 0--1 Matrices.}
If $S$ is a sequence, let $|S|$ be its length
and $\|S\|$ the size of its alphabet $|\Sigma(S)|$.
A \emph{block} is a contiguous sequence of distinct
symbols.  If $S$ is understood to be partitioned
into blocks, $\bl{S}$ is the number of blocks.
Regardless of $\Sigma(S)$, we can always write 
$S$ in \emph{canonical form} over the alphabet
$\{1,\ldots,\|S\|\}$, where the symbols are sorted
according to their first appearance in $S$.
If $S$ is in canonical form, its canonical 
matrix $A_S$ is the $\|S\| \times \bl{S}$ symbol-block incidence matrix, i.e., 
$A_S(i,j) = 1$ if symbol $i$ appears in block $j$, and 0 otherwise.  
One cannot quite recover $S$ from $A_S$ since 
$A_S$ does not encode the order of symbols within a block.  
Nonetheless, the transformation is \emph{useful} inasmuch
as subsequences avoided by $S$ often become 
0--1 patterns avoided by $A_S$.

\paragraph{Composition and Shuffling.}
We consider sequences $S$ partitioned into 
\emph{live} and \emph{dead} blocks satisfying extra constraints:
\begin{itemize}
    \item All live blocks have the same length.  Dead blocks have variable lengths, 
    	and the number of dead blocks between consecutive live blocks is also variable.

    \item The first occurrence of every symbol appears in a dead block,
    and dead blocks contain only first occurrences.
    Let $\livebl{S}$ be the number of live blocks in $S$.
\end{itemize}

\paragraph{Composition.}
Suppose $\Utop$ is a sequence in which all live blocks have length $j$ and $\Umid$ is a sequence with $\|\Umid\|=j$.
The \emph{composition} $\Usub = \Utop\compose\Umid$ is obtained by replacing each live block $L$ of $\Utop$ 
with a copy $\Umid(L)$ over the alphabet of $L$,
whereas dead blocks of $\Utop$ are inherited by $\Usub$ verbatim.
In general $\Umid(L)$ can contain both live and dead blocks~\cite{Pettie15-SIDMA}, 
but in our particular construction $\Umid(L)$ contains only live blocks.

\paragraph{Shuffling.}
Now suppose $\Usub$ is a sequence whose live blocks
have length $j$ and $\Ubot$ is a sequence with $\livebl{\Ubot}=j$. 
The \emph{left-shuffle} $\Usub\lshuffle\Ubot$ is obtained as follows.  
Let 
$\Usub = D_0 L_1 D_1 L_2 D_2 \cdots L_k D_k$, 
where $L_i$ is the $i$th live block, 
$D_i$ is zero or more dead blocks,
and $k=\livebl{\Usub}$.  
Let $\Ubot^* = \Ubot^{(1)}\cdots \Ubot^{(k)}$ 
be the concatenation of $k$ copies of 
$\Ubot$ over disjoint alphabets.
The \emph{left-shuffle} is obtained by taking, for all $i$, the block $L_i = (a_1 a_2 \cdots a_{j})$ and inserting $a_\ell$, for $\ell\in[1,j]$, 
at the 
\emph{left end} of the 
$\ell$th live block of $\Ubot^{(i)}$,
then inserting dead blocks $D_i$ between $\Ubot^{(i)}$ and $\Ubot^{(i+1)}$.\footnote{The \emph{right shuffle} 
$\Usub\rshuffle\Ubot$ 
is defined in the same way, except that $a_\ell$ is inserted at the \emph{right end} 
of the $\ell$th live block of $\Ubot^{(i)}$.
We only use left-shuffles but there are cases where it is desirable to use both~\cite{Pettie-GenDS11,Pettie15-SIDMA}.}

\paragraph{Sequence Construction.}

$U(j)$ and $U(i,j)$ are blocked sequences, where square brackets
indicate dead blocks and parentheses indicate live blocks.
$U(i,j)$ is a variation on order-4 Davenport-Schinzel sequences~\cite{ASS89}, adapted specifically to exclude a 
small pattern that arises from $P\Kron\hatpattern$-type patterns.

\begin{align*}
U(j) &= (j\, (j-1)\, \cdots \, 1)(1\, 2\, \cdots \, j)					& \mbox{2 live blocks}\\
U(1,j) &= [1\, 2\, \cdots \, j] (1\, 2\, \cdots \, j)					& \mbox{first block dead, second live}\\
U(i,0) &= (\,)^{2}							& \mbox{2 empty live blocks}\\
U(i,j) &= (\Utop\compose\Umid)\lshuffle \Ubot\\
\mbox{where} &\hcm[1] \Utop = U(i-1, \livebl{U(i,j-1)}),\\
	        &\hcm[1] \Umid = U(\livebl{U(i,j-1)}),\\
\mbox{and}	&\hcm[1] \Ubot = U(i,j-1)
\end{align*}

Let $N(i,j) = \|U(i,j)\|$ be the alphabet size and $L(i,j) = \livebl{U(i,j)}$ be the number of live blocks. $N,L$ obey the following recurrence:
\begin{align*}
L(1,j) &= 1\\
L(i,0) &= 2\\
L(i,j) &= L(i,j-1)\cdot 2\cdot L(i-1,L(i,j-1))\\
N(1,j) &= j\\
N(i,j) &= N(i,j-1) \cdot 2\cdot L(i-1,L(i,j-1)) + N(i-1,L(i,j-1))
\end{align*}

\begin{lemma}\label{lem:U-Ackermann}
Fix any $U=U(i,j)$.
\begin{enumerate}
\item All live blocks in $U$ have length $j$. 
Each symbol appears $2^{i-1}+1$ times in $U$, 
its first occurrence appearing in a dead block, and the remaining $2^{i-1}$ times in live blocks.
\item As a consequence of part 1, $N(i,j) = (j/2^{i-1}) \cdot L(i,j)$. 
\item The number of dead blocks is at most $L(i,j)-1$.
\item If $n=N(i,j)$ is the number of symbols in $U$ and $m<2L(i,j)$ the number of blocks, $i=\alpha(n,m)\pm O(1)$, 
and $|U| = \Theta(n2^{\alpha(n,m)})$.
\end{enumerate}
\end{lemma}

\begin{proof}
\underline{Part 1.} The claim holds in the base cases $U(1,j)$ and $U(i,0)$.
All live blocks in $\Ubot=U(i,j-1)$ have length $j-1$ by induction, and each receive one symbol in the shuffling operation $(\Utop\compose\Umid)\lshuffle\Ubot$.  
All symbols in $\Ubot=U(i,j-1)$ appear in $2^{i-1}$ live blocks.
Those in $\Utop=U(i-1,\livebl{U(i,j-1)})$ appear in $2^{i-2}$ live blocks, and therefore in $2^{i-1}$ live blocks in $\Usub=\Utop\compose\Umid$ since $\Umid$ doubles the number of live occurrences.  The property that first occurrences appear in dead blocks is preserved by composition and shuffling.
\underline{Part 2.} Note that both 
$j\cdot L(i,j)$ and $2^{i-1}\cdot N(i,j)$ 
both count the total length of all live blocks.
\underline{Part 3.} The claim holds in all base cases.
By induction, the number of dead blocks in $\Utop$
is at most $L(i-1,L(i,j-1))-1$.
$\Ubot^*$ consists of $2L(i-1,L(i,j-1))$ copies of $\Ubot=U(i,j-1)$,
so $\Ubot^*$ has $2L(i-1,L(i,j-1))(L(i,j-1)-1)$ dead blocks.
In total there are $L(i-1,L(i,j-1))(2L(i,j-1)-1)-1\leq L(i,j)-1$ dead blocks.
\underline{Part 4.} Proving Ackermann-like functions are equivalent inasmuch as their inverses differ by $\pm O(1)$ is tedious, 
but straightforward. See~\cite[Lemma 3.10]{PetInvAck} for an example 
of such a proof.
\end{proof}

\begin{lemma}\label{lem:U-properties}
Let $U=U(i,j)$ be obtained from $\Utop,\Umid,\Ubot$.
Suppose $a<b$ are two symbols in $\Sigma(U)$ appearing in a common live block.
\begin{enumerate}
\item The restriction of $U$ to letters $\{a,b\}$
is of the form $a^* b^* (ab) b^* a^*$.
\item If $a\in \Sigma(\Utop)$, $b\in \Sigma(\Ubot^*)$,
then $a<c$ for \emph{every} symbol $c$ appearing in $b$'s copy of $\Ubot$.
\end{enumerate}
\end{lemma}

\begin{proof}
The claim is true in the base cases $U(1,j)$ and $U(i,0)$.
Consider the moment that $a$ is shuffled into $b$'s live block, where $a\in\Sigma(\Utop)$ and $b\in\Sigma(\Ubot^*)$.
All occurrences of $b$ appear in one copy of $\Ubot$ in $\Ubot^*$, and exactly one occurrence
of $a$ is shuffled into this copy.  It follows that the restriction of $U$ to letters $\{a,b\}$
is of the form $a^* | b^* (ab) b^* | a^*$, where the bars mark the boundary of $b$'s copy of $\Ubot$.
Furthermore, since the first occurrence of $a$ is in a dead block, 
which is inserted between two copies of $\Ubot$ in $\Ubot^*$, 
$a<c$ for \emph{every} $c$ in $b$'s copy of $\Ubot$.
\end{proof}

\begin{lemma}\label{lem:41213}
$U=U(i,j)$ does not contain any subsequences order-isomorphic to  $41213$.
\end{lemma}

\begin{proof}
Since $U$ is in canonical form, the existence of $41213$ 
implies the existence
of a subsequence order-isomorphic to
\[
\sigma = 31213
\]
Suppose that $\sigma$ first appears in $U(i,j)=\Usub\lshuffle\Ubot = (\Utop\compose\Umid)\lshuffle\Ubot$.
If $\sigma$ already appears in $\Usub$ but did not appear in $\Utop$,
then $\Utop$ must have contained $\sigma'$.
\begin{align*}
\sigma' &= 31(12)3
\end{align*}
Note that $\{2,3\}$ cannot share a live block in $\Utop$ 
without also including $1$,
and if $\{1,2,3\}$ shared a live block in $\Utop$, 
the restriction of $\Utop$ to $\{1,2,3\}$ would, by Lemma~\ref{lem:U-properties}(1),
be:
\[
1^*2^*3^*(123)3^*2^*1^*
\]
and the restriction of $\Usub$ to $\{1,2,3\}$ would be:
\[
1^*2^*3^*(321)(123)3^*2^*1^*,
\]
which does not contain $\sigma$.
We therefore need to argue that neither 
$\sigma$ nor $\sigma'$ can arise in $U(i,j)$ 
in the \emph{shuffling} operation.

If $\sigma$ or $\sigma'$ arose during shuffling, 
then Lemma~\ref{lem:U-properties} implies that 
for any $a,b\in \{1,2,3\}$ with $a\in \Sigma(\Utop)$ and $b\in\Sigma(\Ubot^*)$, that $a<b$.
It cannot be that $\{1\}$ or $\{1,2\}\subset \Sigma(\Utop)$ 
while $\{2,3\}$ or $\{3\}\subset \Sigma(\Ubot^*)$
since 3's 
copy of $\Sbot$ only receives \emph{one} copy of any symbol during shuffling, but both $\sigma,\sigma'$ have two 1s between the first and last 3.
\end{proof}

\begin{remark}
The distinction between \emph{live} and \emph{dead} blocks
is critical for constructing order-3 ($ababa$-free) Davenport-Schinzel sequences~\cite{HS86,Nivasch10,Pettie-DS-JACM}, 
and generalized DS sequences with length $O(n\poly(\alpha(n)))$~\cite{Pettie-GenDS11,Pettie15-SIDMA}. However, all constructions of DS sequences at 
order-4 and above~\cite{ASS89,Nivasch10,Pettie-DS-JACM}
(having length $\Omega(n2^{\alpha(n)})$)
only use live blocks.  
In Lemma~\ref{lem:41213},
it is very important that \emph{first} occurrences lie exclusively in dead blocks, 
and are never shuffled into the 
middle of a copy of $\Ubot$.
If the first block in $U(1,j)$ were redefined to be live,
then we would see instances 
of $41213$ in $U(i,j)$.  It could be that 
$1\, 2\, 1\, 2\, 1$ appears in a copy of $\Ubot$,
and the first occurrences of $\{3,4\}$ 
lie in a common live block in $\Utop$. 
The restriction of $\Utop$ to $\{3,4\}$ 
contains $(34) 3$.  After shuffling the block $(34)$ into the $\Ubot$ containing $1,2$ we can see $1\,2\,3\,4\, 1\,2\,1\,3$.
Lemma~\ref{lem:U-properties}(2) rules out 
this possibility when first occurrences 
appear in \emph{dead} blocks.
\end{remark}

\begin{theorem}\label{thm:main-lower-bound}
$\Ex(W,n,m) = \Theta(m+n2^{\alpha(n,m)})$, where
\[
W = \scalebox{.7}{$\left(\begin{array}{ccccccc}
\bu &&&&\\
&&&&\bu\\
&&\bu&&\\
&\bu&&\bu&
\end{array}
\right)$}
\]
\end{theorem}

\begin{proof}
By Lemma~\ref{lem:U-Ackermann}, the sequence $U=U(i,j)$ has $n=N(i,j)$ symbols, 
$m<2L(i,j)$ blocks, and length $|U| = \Theta(n2^i) = \Theta(n2^{\alpha(n,m)})$.
We convert $U$ to an $n\times m$ 0--1 matrix $A_U$.  Number the rows of $A_U$
from bottom-to-top, and the columns from left-to-right, and let 
$A_U(i,j)=1$ iff symbol $i$ appears in block $j$.
$U$ does not contain subsequences order-isomorphic to $41213$,
which implies that $A_U$ is $W$-free, and hence $\Ex(W,n,m)=\Omega(n2^{\alpha(n,m)})$.
The matching upper bound is obtained as in \cite[Thm.~3.4]{Pettie-GenDS11}.
\end{proof}

\section{Additional Upper Bounds}\label{sect:additional-upper-bounds}

\subsection{Proof of Theorem~\ref{thm:identity-product}}

Recall that $I_k$ is the $k\times k$ identity matrix. 
For example, when $k=3$, we have the following.
\begin{align*}
    I_k &= \scalebox{.8}{$\left(\begin{array}{ccc}
    \bu&&\\
    &\bu&\\
    &&\bu
    \end{array}\right)$}
    &
    I_k \Kron \hatpattern &= \scalebox{.8}{$\left(\begin{array}{c@{\hcm[.1]}c@{\hcm[.1]}c@{\hcm[.2]}c@{\hcm[.1]}c@{\hcm[.1]}c@{\hcm[.4]}c@{\hcm[.1]}c@{\hcm[.1]}c}
      &\bu&&&&&&&\\
    \bu&&\bu&&&&&&\\ 
    &&&&\bu&&&&\\
    &&&\bu&&\bu&&&\\
    &&&&&&&\bu&\\
    &&&&&&\bu&&\bu
    \end{array}\right)$}
\end{align*}

Theorem~\ref{thm:identity-product} follows from Keszegh's~\cite{Keszegh09} \emph{joining} operation and
and Pettie's upper bound on order-3 Davenport-Schinzel sequences~\cite{Pettie-DS-JACM}; cf.~\cite{HS86,Nivasch10}.
Keszegh~\cite{Keszegh09} proved that if $R$ has a 1 in its southeast corner and $S$ has a 1 in its northwest corner, that 
$\Ex(R\oplus S, n,m) \leq \Ex(R,n,m)+\Ex(S,n,m)$, 
where $R\oplus S$ is formed by joining $R,S$ at their corners.
\[
R\oplus S = \left(\,\begin{array}{ccccccc}
\cline{1-4}
\multicolumn{1}{|c}{ } &&& \multicolumn{1}{c|}{} \\
\multicolumn{1}{|c}{ } &R&& \multicolumn{1}{c|}{} &&& \\\cline{4-7}
\multicolumn{1}{|c}{ } &&\zero{\hcm[.26]{\bu}}& \multicolumn{1}{|c|}{} &&& \multicolumn{1}{c|}{}\\\cline{1-4}
&&& \multicolumn{1}{|c}{} &&S& \multicolumn{1}{c|}{}\\
&&& \multicolumn{1}{|c}{} &&& \multicolumn{1}{c|}{}\\\cline{4-7}
\end{array}\,\right).
\]
Observe that $I_k\Kron \hatpattern = (I_{k-1}\Kron \hatpattern)\oplus \ababa$, so we can apply Keszegh's operation $k-2$ times to reduce to the base case $I_2\Kron\hatpattern$.  We claim $\Ex(I_2\Kron\hatpattern,n,m) < \Ex(\ababa,n,m) +2n +m$.  
Suppose $A$ is $I_2\Kron\hatpattern$-free, and let 
$A'$ be obtained by removing
the top 1 in each column, and then the first two 1s in each row.
Then $A'$ is clearly $\ababa$-free.  
Putting it all together, we have,
\begin{align*}
\Ex(I_k\Kron\hatpattern,n,m) &\leq \Ex(I_2\Kron\hatpattern,n,m) + (k-2)\Ex(\ababa,n,m) \\
&\leq (k-1)\Ex(\ababa,n,m) + 2n +m\\
&\leq (k-1)2n\alpha(n,m)+O(k(n+m)),
\end{align*}
where the last inequality follows from the bound $\Ex(\ababa,n,m)=2n\alpha(n,m)+O(n+m)$ on order-3 Davenport-Schinzel sequences~\cite{Pettie-DS-JACM}.

\subsection{Avoding $W$ and Its Reflection}

By symmetry, Theorem~\ref{thm:main-lower-bound} also applies
to $\Ex(W',n,m)$, where $W'$ is the reflection of $W$ along the $y$-axis.  However, the density of $\{W,W'\}$-free matrices 
is asymptotically smaller.

\begin{theorem}
$\Ex(\{W,W'\},n,m) \leq 4n\alpha(n,m)+O(n+m)$, 
where $W'$ is the reflection of $W$ along the $y$-axis.
\[
W' = \scalebox{.7}{$\left(\begin{array}{ccccc}
&&&&\bu\\
\bu&&&&\\
&&\bu&&\\
&\bu&&\bu&
\end{array}
\right)$}
\]
\end{theorem}

\begin{proof}
Let $A$ be a $\{W,W'\}$-free matrix.
Remove the top $1$ in each column, yielding
$A'$.  It follows that $A'$ is $\{W,W',W''\}$-free,
where
\[
W'' = \scalebox{.7}{$\left(\begin{array}{ccccc}
\bu&&&&\bu\\
&&\bu&&\\
&\bu&&\bu&
\end{array}
\right)$}
\]
We prove that 
$\Ex(\{W,W',W''\},n,m)\leq 2\Ex(\ababa,n,m)$.
Call a 1 in $A'$ \emph{bottom-right} if it appears 
as the bottom-right 1 in a copy of $\ababa$,
and \emph{bottom-left} if it appears as the bottom-left 1 in a copy of $\ababareflect$.
If $\|A'\|_1 \geq 2\Ex(\ababa,n,m)+1$ then some
$A'(i,j)=1$ 
must be classified as both bottom-left and bottom-right.
Let $(i_L,j_L)$ and $(i_R,j_R)$ be the positions
of the top-left 1 in a copy of $\ababa$ containing $A'(i,j)$
and top-right 1 in a copy of $\ababareflect$ containing $A'(i,j)$, respectively.  Numbering the rows from bottom to top,
we have a copy of $W''$ if $i_L=i_R$,
a copy of $W'$ if $i_L<i_R$, 
and a copy of
$W$ if $i_L>i_R$.  For example, when $i_L < i_R$,
\[
W'' = \scalebox{.7}{$\left(\begin{array}{ccccccc}
&&&&&&\bu\\
&&&&&&\text{\footnotesize{$(i_R,j_R)$}}\\
&&&&\bu&&\\
\bu&&&&&&\\
\text{\footnotesize{$(i_L,j_L)$}}&&&&&&\\
&&\bu&&&&\\
&\bu&&\bu&&\bu&\\
&&&\text{\footnotesize{$(i,j)$}}&&&
\end{array}
\right)$}
\]
It is known~\cite{Pettie-DS-JACM,Nivasch10} 
that $\Ex(\ababa,n,m)=2n\alpha(n,m)\pm O(n+m)$.
\end{proof}

\section{Concluding Remarks}\label{sect:conclusion}

Demaine et al.\cite{DemaineHIKP09} and Pettie~\cite{Pettie10a} introduced the idea
of representing the behavior of a data 
structure by a 0--1 matrix in which the axes 
correspond to \emph{time} and data structure \emph{elements}.
Applying results on forbidden 0--1 matrices
in this context is very natural,
and has led to sharp or nearly sharp bounds
on certain path compression schemes~\cite{Pettie10a},
structured inputs to binary search trees~\cite{ChalermsookGKMS15,ChalermsookGKMS15b,ChalermsookGJAPY23,Pettie10a},
and heaps~\cite{KozmaS20}.

\medskip 

Although the forbidden 0--1 matrix framework is 
perfectly suited to proving that sorting 
$\pi$-free sequences takes \emph{near-linear} time,
it might be inadequate to establishing 
an optimal $O_k(n)$ time bound.  As we have shown,
any $o(n2^{\alpha(n)})$-time analysis must use
some property \emph{beyond} 
$P_{\pi}\Kron\hatpattern$-freeness.
Here it may be useful to consider different ways
to decompose a 0--1 matrix; see Guillemot and Marx~\cite{GuillemotM14} 
and Chalermsook, Gupta, Jiamjitrak, Acosta, Pareek, and Yinghcareonthawornchai~\cite{ChalermsookGJAPY23}.

\medskip

The literature on forbidden 0--1 matrices is rich~\cite{FurediH92,Tardos05,PachTardos06,Pettie-GenDS11,Pettie-SoCG11,Pettie-FH11,Pettie15-SIDMA,Keszegh09,Geneson09,Fulek09,KorandiTTW19,CibulkaK12,CibulkaK17,Fox13,MarcusT04,Klazar92,PettieT23} 
but there are many outstanding open problems.
In the context of \emph{data structure analysis}, the most
interesting open problems are to characterize the set of \emph{linear}
forbidden patterns---those $P$ with $\Ex(P,n)=O(n)$---and in particular,
to characterize linear \emph{light} patterns.
It is known that there are infinitely many minimal (with respect to $\prec$)
non-linear patterns~\cite{Keszegh09,Geneson09,Pettie-FH11}, 
but there may be other ways to characterize this set in a finite representation.
On the other hand, we know of only two minimally non-linear \emph{light} patterns 
(with respect to $\prec$ and reflections), namely
$\scalebox{.5}{$\left(\begin{array}{cccc}
    \bu&&&\\
    &&\bu&\\
    &\bu&&\bu
    \end{array}\right)$}$
and 
$\scalebox{.5}{$\left(\begin{array}{cccc}
    \bu&&\bu&\\
    &\bu&&\bu
    \end{array}\right)$}$.  
It is quite possible that these are the only sources 
    of non-linearity in light patterns.

\newcommand{\etalchar}[1]{$^{#1}$}

\end{document}